\documentclass[11pt]{article}
\usepackage{cite}
\usepackage{graphics,amsmath,amssymb,amsthm}
\newtheorem{theorem}{Theorem}[section]
\newtheorem{proposition}{Proposition}[section]

\newcommand{\be}{\begin{equation}}
\newcommand{\ee}{\end{equation}}
\newtheorem{lemma}{Lemma}[section]
\numberwithin{equation}{section}
\def \K{[\![}
\def \J{]\!]}

\def \ui#1{^{(#1)}}

\def \ui#1{^{(#1)}}

\title{Symmetries  for the Ablowitz-Ladik hierarchy: I. \\Four-potential case
}
\author {Da-jun Zhang\footnote{Corresponding author. E-mail: djzhang@staff.shu.edu.cn},~~Shou-ting Chen
\\
{\small\it Department of Mathematics, Shanghai
University, Shanghai 200444, P.R. China}}
\date{}

\begin{document}
\maketitle

\begin{abstract}
In the paper we first investigate symmetries of  isospectral and
non-isospectral four-potential Ablowitz-Ladik hierarchies. We
express these hierarchies in the form of $u_{n,t}=L^m H^{(0)}$,
where $m$ is an arbitrary integer (instead of a nature number) and
$L$ is the recursion operator. Then by means of the zero-curvature
representations of the isospectral and non-isospectral flows, we
construct symmetries for the isospectral equation hierarchy as well
as non-isospectral equation hierarchy, respectively. The symmetries,
respectively, form two centerless Kac-Moody-Virasoro algebras. The
recursion operator $L$ is proved to be hereditary and a strong
symmetry for this isospectral equation  hierarchy. Besides, we make
clear for the relation between four-potential and two-potential
Ablowitz-Ladik hierarchies. The even order members in the
four-potential Ablowitz-Ladik hierarchies together with their
symmetries and algebraic structures can be reduced to two-potential
case. The reduction keeps invariant for the algebraic structures and
the recursion operator for two potential case becomes $L^2$.

\vskip 5pt

\end{abstract}

\section{Introduction}
As is well known, many physically interesting problems are modeled
by nonlinear differential-difference equations, such as the Toda
lattice, Volterra lattice and discrete nonlinear Schr\"odinger
equation. Since 1970s discrete systems have received considerable
attention from variety of aspects (e.g.,
\cite{AL-75-JMP,Hirota-1970s,Date-1980s,Kuper-book}), such as
Inverse Scattering Transform, bilinear method, Sato's approach,
symmetry analysis and so on. One of famous discrete spectral
problems is given by Ablowitz and Ladik\cite{AL-75-JMP,AL-76-JMP}
which is now referred to as the Ablowitz-Ladik (AL) spectral
problem. This spectral problem, coupled with different time
evolution parts, has provided Lax integrabilities for many discrete
soliton systems, such as integrable discrete nonlinear Schr\"odinger
equation\cite{Ablowitz-04-book}, discrete mKdV equation and so
forth.

There are two types of the AL spectral problems,
which contains two potentials $\{Q_n,R_n\}$ and four potentials $\{Q_n,R_n,S_n,T_n\}$, respectively.
The two-potential one is the direct discretization (cf. \cite{Ablowitz-04-book}) of the famous continuous AKNS-ZS spectral problem\cite{AKNS},
and besides solutions, the related Hamiltonian structures, constraint flows, nonlinearization,
Darboux transformation, conservation laws, symmetries and Lie algebra structures
have been studied (cf. \cite{Geng-1989,Zeng-95-JPA,Ma-Tamizhimani-JPSJ-1999,ZDJ-02-JPA,ZDJ-02-CSF,Geng-03-JMP,Veks-2006,ZDJ-06-PLA,Geng-07-SAM,Gesztesy-08-SAM}).
The four-potential AL spectral problem is more complicated than the two-potential case
because of containing two more potentials and its unsymmetrical matrix form.
Cheng \cite{CY-86} transformed the four-potential AL spectral problem to a bit simple form
which could be further related to the two-potential case, but the relation (connecting two and four potentials) given in \cite{CY-86} is nothing helpful for
discussions of Hamiltonian structures and symmetries.
Recently, based on Cheng's transformation, Geng and Dai \cite{Geng-4p} separated the four-potential
AL spectral matrix into two symmetrical two-potential AL spectral matrices.
This makes it possible to construct a four-potential hierarchy (e.g.\cite{Geng-4p}) and its recursion operator
and then investigate more characteristics of integrability.

Infinite symmetries act as an important characteristic for
integrable systems \cite{Fokas-1987}. Symmetry-analysis is also a
powerful approach to finding exact solutions for nonlinear systems
\cite{Olver-book,Bluman-book}. In this paper we  focus on symmetries
and their Lie algebras for the four-potential isospectral and
non-isospectral AL hierarchies. We will first derive positive and
negative order isospectral and non-isospectral flows. These flows
share a same recursion operator $L$ and can be respectively
uniformed as $K^{(m)}=L^m K^{(0)}$ and $\sigma^{(m)}=L^m
\sigma^{(0)}$, where $m$ is an arbitrary integer (instead of a
nature number in most of cases). Then we can imbed these flows into
their zero-curvature equations by means of functional derivatives.
The resulting expressions, which we refer to as zero-curvature
representations, have been shown to be powerful in constructing
symmetries for Lax integrable systems
(cf.\cite{Ma-Tamizhimani-JPSJ-1999,ZDJ-06-PLA,Ma-1990,Chen-1991,Chen-1996,Chen-2003,MWX-99-JMP}).
We will derive algebraic relations for isospectral and
non-isospectral flows and then derive symmetries and their Lie
algebras for not only isospectral AL hierarchy but also
non-isospectral hierarchy. Both algebras are the type of centerless
Kac-Moody-Virasoro algebra. The recursion operator $L$ is hereditary
and a strong symmetry for the isospectral hierarchy.

A natural question is whether all these structures w.r.t. four potentials $(Q_n,R_n,S_n,T_n)$,
including hierarchies, recursion operator $L$, symmetries and algebras admit a closed reduction
w.r.t. two potentials $(Q_n,R_n)$ by directly taking $(S_n,T_n)=(0,0)$.
This is true for those even order members in the
four-potential isospectral and non-isospectral hierarchies,
and the new recursion operator becomes $L^2$.
We will discuss the reduction in the paper.

This paper is the first part of our series investigations which
consist of two parts. In Part II we will focus on symmetries of the
integrable discrete nonlinear Schr\"odinger equation and discrete
AKNS hierarchy. The integrable discrete nonlinear Schr\"odinger
equation consists of positive and negative order flows which
correspond to a central-difference discretization for a continuous
second order derivative. We will also give a recursion operator
which generates discrete AKNS hierarchies. The obtained symmetry
algebras are not centerless Kac-Moody-Virasoro type. The structure
changes will also be explained in Part II.

The present paper is organized as follows.
Sec.2 contains some basic notations and backgrounds on the AL spectral problem.
Sec.3 derives four-potential isospectral and non-isospectral flows and their zero-curvature representations.
In Sec.4 we derive symmetries and their algebraic structures
for both four-potential isospectral and non-isospectral hierarchies.
In Sec.5 we discuss reduction relation between four-potential case and two-potential case.
There are also two Appendix sections.
Sec.A lists out first few equations in four-potential AL hierarchies and their Lax pairs,
and Sec.B is a theorem obtained in Ref.\cite{ZDJ-02-JPA} which we give here for self-containedness.

\section{Basic notations and backgrounds}

Let us first introduce some basic notations and notions which have been used
for discussing symmetries of discrete systems (cf. \cite{MWX-99-JMP,ZDJ-02-JPA,ZDJ-06-PLA}).

Assume that $u_n\doteq u(t, n)=(u^{(1)}, u^{(2)}, u^{(3)},u^{(4)})^T$ is a four-dimensional vector field, where
$u^{(i)}=u^{(i)}(t, n),~1\leq i\leq 4$, are all functions defined
over $\mathbb{R}\times \mathbb{Z}$ and vanish rapidly as
$|n|\rightarrow \infty$.
By $\mathcal{V}_{4}$ we denote a linear space consisting of all vector
fields $f=(f^{(1)}, f^{(2)}, f^{(3)}, f^{(4)})^{T}$, where each $f^{(i)}$
is a function of $u(t, n)$ and its shifts $u(t, n+j),~j\in \mathbb{Z}$,   satisfying $f^{(i)}(u(t, n))|_{u_{n}=0}=0$,
and each $f^{(i)}$
is $C^{\infty}$ differentiable w.r.t. $t$ and $n$, and
$C^{\infty}$-Gauteaux differentiable w.r.t. $u_{n}$.
Here the Gateaux (or Fr\'echet) derivative of $f\in
\mathcal{V}_{4}$ (or $f$ an  operator  on $\mathcal{V}_{4})$ in the
direction $g\in \mathcal{V}_{4}$ is defined as
\begin{equation}
f^{\prime}[g]=\frac{d}{d \epsilon}\Bigr|_{\epsilon=0}f(u+\epsilon g).
\label{def-gat}
\end{equation}
By means of  the Gateaux derivative one can define a Lie product for any $f,
g\in \mathcal{V}_{4}$  as
\begin{equation}
\K  f, g\J  =f^{\prime}[g]-g^{\prime}[f].
\end{equation}

We also define  a Laurent matrix polynomials
space $\mathcal {Q}_{2}(z)$ composed by all $2\times 2$ matrices $Q=Q(z, u(t,
n))=(q_{ij}(z, u(t, n)))_{2\times 2}$, where all the $\{q_{ij}\}$    are
Laurent polynomials of $z$. Two subspaces of $\mathcal
{Q}_{2}(z)$ we will need are
\begin{align*}
&\mathcal {Q}^{+}_{2}(z)=\{Q\in \mathcal {Q}_{2}(z)|\mathrm{~the~ lowest~degree~of}~ z\geq 0\},\\
&\mathcal {Q}^{-}_{2}(z)=\{Q\in \mathcal {Q}_{2}(z)|\mathrm{~the~ highest~degree~ of}~ z\leq 0\}.
\end{align*}

We note that in a similar way we can define spaces $\mathcal{V}_{s}$, $\mathcal{Q}_{m}(z)$ and
$\mathcal {Q}^{\pm}_{m}(z)$ (cf. \cite{ZDJ-02-JPA}).

In general a discrete evolution equation arises from
the compatibility of a pair discrete linear problems\footnote{
Actually, they are semi-discrete.}
\begin{equation}
\Phi_{n+1}=U_{n}(z, u(t, n))\Phi_n,~~~~\Phi_{n,t}=V_{n}(z, u(t, n))\Phi_n,
\end{equation}
where $\Phi_n$ is a wave function, $U_{n}$ is a spectral matrix with
spectral parameter $z$ and potential vector $u(t, n)$ while $V_n$ is a matrix
governing time evolution.
The compatibility condition, also called  discrete zero-curvature equation, reads
\begin{equation}
\label{2}
U_{n, t}=(EV_{n})U_{n}-U_{n}V_{n}.
\end{equation}
Here and in the following $E$ is a shift operator defined as $E^jf(n)=f(n+j)$ for $j\in \mathbb{Z}$.
Suppose that the corresponding nonlinear evolution equation is
\begin{equation}
\label{44.1}
u_{n, t}=K(u_{n}).
\end{equation}
Then by means of the Gateaux derivative  the flow $K(u_{n})$ can be embedded
into the zero-curvature equation \eqref{2} as the following,
\begin{equation}
\label{3}
U_{n}^{\prime}[K(u_{n})]=(EV_{n})U_{n}-U_{n}V_{n}-U_{n, z}z_{t},
\end{equation}
which is usually called  the zero-curvature representation of the flow $K(u_{n})$.

For the nonlinear evolution equation \eqref{44.1}, $\sigma(u_{n})\in \mathcal{V}_{4}$
is its symmetry if $\sigma_{t}=K^{\prime}[\sigma]$, i.e.,
\begin{equation}
\label{44.11}
\frac{\tilde{\partial} \sigma}{\tilde{\partial} t}=\K  K, \sigma\J  ,
\end{equation}
where by $\frac{\tilde{\partial}\sigma}{\tilde{\partial} t}$
we specially denote the derivative of $\sigma$ w.r.t. $t$ explicitly included in  $\sigma$,
(for example, $\frac{\tilde{\partial}\sigma }{\tilde{\partial} t}=u_{n}$ if $\sigma=t u_n+u_{n+1}$).

Next, let us recall some backgrounds on four-potential AL hierarchy.
The four-potential AL spectral problem reads
\cite{AL-75-JMP,AL-76-JMP}
\begin{equation}
\label{4p-AL}
\begin{array}{l}
\psi_{1, n+1}=\lambda\psi_{1, n}+Q_{n}\psi_{2, n}+S_{n}\psi_{2, n+1},\\
\psi_{2,n+1}=\lambda^{-1}\psi_{2, n}+R_{n}\psi_{1, n}+T_{n}\psi_{1, n+1},
\end{array}
\end{equation}
where $\lambda$ is a spectral parameter and $Q_{n}, R_{n}, S_{n}, T_{n}$ are four potential functions of $n$ and $t$.
When $S_n=T_n=0$ \eqref{4p-AL} reduces to the two-potential AL spectral problem, i.e.,
\begin{equation}
\label{2p-AL}
\begin{array}{l}
\psi_{1, n+1}=\lambda\psi_{1, n}+Q_{n}\psi_{2, n},\\
\psi_{2,n+1}=\lambda^{-1}\psi_{2, n}+R_{n}\psi_{1, n},
\end{array}
\end{equation}
which is a discrete version of the AKNS-ZS spectral problem (cf.\cite{Ablowitz-04-book}).

An alternative (matrix) form of \eqref{4p-AL} is\cite{AL-75-JMP}
\begin{equation}\label{5}
\Psi_{n+1}=\frac{1}{\Lambda_{n}}\left(
\begin{array}{cc} z^2+S_{n}R_{n}& Q_{n}+z^{-2}S_{n}\\
z^2T_{n}+R_{n}&z^{-2}+T_{n}Q_{n}
\end{array}\right)\Psi_{n},~~~~\Psi_{n}=\left(
\begin{array}{cc} \psi_{1,n}\\
\psi_{2,n}
\end{array}\right),
\end{equation}
where $\Lambda_{n}=1-S_{n}T_{n}$ and we have substituted $z^2$ for $\lambda$.
This form can be gauge-transformed to \cite{CY-86}
\begin{equation}
\Phi_{n+1}=U_n\Phi_{n},~~~
U_n=\left(
\begin{array}{cc} z^{2}+S_{n}R_{n}& zQ_{n}+z^{-1}S_{n}\\
zT_{n}+z^{-1}R_{n}&z^{-2}+T_{n}Q_{n}
\end{array}\right),~~~\Phi_{n}=\left(
\begin{array}{cc} \phi_{1,n}\\
\phi_{2,n}
\end{array}\right),
\label{4p-new}
\end{equation}
where $\Phi_{n}$ and $\Psi_n$ are related through
\begin{equation}
\label{6}
\Phi_{n}=\rho(z)\Psi_{n}\prod^{+\infty}_{i=n}\Lambda_{i}^{-1},
~~~~\rho(z)=\left(
\begin{array}{cc} z^{\frac{1}{2}}& 0\\
0&z^{-\frac{1}{2}}
\end{array}\right).
\end{equation}
Here on $U_n$ we impose a condition
\begin{equation}
\label{cd}
Q_nR_n+S_nT_n\neq 0
\end{equation}
so that $U_n'$ is an injective homomorphism.

Suppose that the time evolution of $\Phi_n$ is
\begin{equation}
\label{4p-new-time}
\Phi_{n,t}=V_{n}\Phi_{n}, \qquad V_n=\left(
\begin{array}{cc} A_{n}& B_{n}\\
C_{n}& D_{n}
\end{array}\right).
\end{equation}
Then the compatibility condition with \eqref{4p-new} yields
\begin{equation}
\label{zce}
U_{n, t}=(EV_{n})U_{n}-U_{n}V_{n}.
\end{equation}
Usually the above discrete zero-curvature equation contributes
a discrete nonlinear evolution equation hierarchy with four potentials and their recursion operator,
but this is not as easy as in two-potential case (related to \eqref{2p-AL}, cf. \cite{Zeng-95-JPA,Ma-Tamizhimani-JPSJ-1999,ZDJ-02-JPA,ZDJ-06-PLA}).
However, the spectral matrix $U_n$ can be separated into \cite{Geng-4p}
\begin{equation}
\label{U-12}
U_{n}=U_{n}^{(2)}U_{n}^{(1)},~~~
U_{n}^{(1)}=\left(
\begin{array}{cc} z & Q_{n}\\
R_{n}& z^{-1}
\end{array}\right),~~
U_{n}^{(2)}=\left(
\begin{array}{cc} z & S_{n}\\
T_{n}& z^{-1}
\end{array}\right).
\end{equation}
Then \eqref{zce} holds if \cite{Geng-4p}
\begin{equation}
\label{zce-aux}
U_{n, t}^{(1)}=\widehat{V}_{n}U_{n}^{(1)}-U_{n}^{(1)}V_{n},\qquad
U_{n, t}^{(2)}=(EV_{n})U_{n}^{(2)}-U_{n}^{(2)}\widehat{V}_{n},
\end{equation}
where $\widehat{V}_{n}= \left(
\begin{array}{cc} a_{n} & b_{n}\\
c_{n}& d_{n}
\end{array}\right)$.
In fact,
\begin{equation*}
U_{n, t}-(EV_{n})U_{n}+U_{n}V_{n} =(U_{n,
t}^{(2)}-(EV_{n})U_{n}^{(2)}+U_{n}^{(2)}\widehat{V}_{n})U_{n}^{(1)}
+U_{n}^{(2)}(U_{n,
t}^{(1)}-\widehat{V}_{n}U_{n}^{(1)}+U_{n}^{(1)}V_{n}).
\end{equation*}
Thus, one can consider the two auxiliary systems given in \eqref{zce-aux},
where each of $U_n^{(j)}$ contains two potentials.
Recently, starting from \eqref{zce-aux} Geng and Dai \cite{Geng-4p}
derived a four-potential AL hierarchy and their recursion relation and considered Hamiltonian structures and nonlinearization of Lax pair.

Noting that in the AL spectral problem \eqref{4p-AL}
the spectral parameter $\lambda$ appears symmetrically
w.r.t. positive and negative powers,
it is then understood that the recursion operator and its inverse can be derived
symmetrically. So are the negative order and positive order hierarchies.
In \cite{ZDJ-06-PLA}, starting from \eqref{2p-AL} we have expressed  isospectral and non-isospectral
two-potential AL hierarchies in the form of $u_{n,t}=L^m H^{(0)}$,
where $m$ is an arbitrary integer (instead of a nature number) and $L$ is the recursion
operator.
This is also true for four-potential case. In the next section we will derive four-potential AL hierarchies
from the two auxiliary systems given in \eqref{zce-aux}.

\section{AL hierarchies and zero-curvature representations}

Now we derive isospectral and non-isospectral AL hierarchies and their recursion operator.
The procedure is quite like the one given in \cite{ZDJ-06-PLA}.
In addition, we will express the obtained isospectral and non-isospectral flows in terms of
zero-curvature equation, by means of which we will prove that the recursion operator is
hereditary and a strong  symmetry of the isospectral hierarchy.

\subsection{Isospectral hierarchy}

The explicit form of the  auxiliary systems \eqref{zce-aux} is
\begin{subequations}
\label{zce-aux-exp}
\begin{align}
z^{-1}z_{t}&= a_{n}-A_{n}+R_{n}b_{n}z^{-1}-C_{n}Q_{n}z^{-1},\label{zce-aux-a}\\
Q_{n, t}&= b_{n}z^{-1}-B_{n}z+Q_{n}(a_{n}-D_{n}),\label{zce-aux-b}\\
R_{n, t}&= c_{n}z-C_{n}z^{-1}+R_{n}(d_{n}-A_{n}),\label{zce-aux-c}\\
z(z^{-1})_{t}&=
d_{n}-D_{n}+Q_{n}c_{n}z-R_{n}B_{n}z;\label{zce-aux-d}
\end{align}
\begin{align}
z^{-1}z_{t}&=  A_{n+1}-a_{n}+T_{n}B_{n+1}z^{-1}-c_{n}S_{n}z^{-1},\label{zce-aux-e}\\
S_{n, t}&= B_{n+1}z^{-1}-b_{n}z+S_{n}(A_{n+1}-d_{n}),\label{zce-aux-f}\\
T_{n, t}&= C_{n+1}z-c_{n}z^{-1}+T_{n}(D_{n+1}-a_{n}),\label{zce-aux-g}\\
z(z^{-1})_{t}&=
D_{n+1}-d_{n}+S_{n}C_{n+1}z-T_{n}b_{n}z.\label{zce-aux-h}
\end{align}
\end{subequations}
From \eqref{zce-aux-a}, \eqref{zce-aux-d}, \eqref{zce-aux-e} and \eqref{zce-aux-h} one can get
\begin{subequations}\label{2.7}
\begin{align}
A_{n}=&(E-1)^{-1}(S_{n}c_{n}z^{-1}+Q_{n}C_{n}z^{-1}-T_{n}B_{n+1}z^{-1}-R_{n}b_{n}z^{-1})+2nz^{-1}z_{t}+A_{0},\\
a_{n}=&
(E-1)^{-1}(S_{n}c_{n}z^{-1}+Q_{n}C_{n}z^{-1}-T_{n}B_{n+1}z^{-1}-R_{n}b_{n}z^{-1})\notag\\
&+Q_{n}C_{n}z^{-1}-R_{n}b_{n}z^{-1}+(2n+1)z^{-1}z_{t}+A_{0},\\
D_n=&(E-1)^{-1}(R_{n}B_{n}z+T_{n}b_{n}z-S_{n}C_{n+1}z-Q_{n}c_{n}z)+2nz(z^{-1})_{t}+D_{0},\\
d_n=&(E-1)^{-1}(R_{n}B_{n}z+T_{n}b_{n}z-S_{n}C_{n+1}z-Q_{n}c_{n}z)+R_{n}B_{n}z-Q_{n}c_{n}z\notag\\
&+(2n+1)z(z^{-1})_{t}+D_{0}.
\end{align}
\end{subequations}
Here $A_{0}=A_{n}|_{u_{n}=0}-2nz^{-1}z_{t}$ and
$D_{0}=D_{n}|_{u_{n}=0}-2nz(z^{-1})_{t}$ where $u_n=(Q_n,R_n,S_n,T_n)^T$. Thus \eqref{zce-aux-exp}
simplifies to
\begin{equation}
\label{u-t}
u_{n, t}=(zL_{1}+z^{-1}L_{2}) \left(
\begin{array}{cccc}
-B_{n}\\ C_{n}\\ -b_{n}\\ c_{n}
\end{array}\right)
+(A_{0}-D_{0}) \left(
\begin{array}{cccc}
Q_{n}\\ -R_{n}\\ S_{n}\\ -T_{n}
\end{array}\right)
+z^{-1}z_{t} \left(
 \begin{array}{cccc}
(4n+1)Q_{n}\\ -(4n+1)R_{n}\\ (4n+3)S_{n}\\ -(4n+3)T_{n}
\end{array}\right),
\end{equation}
where
\begin{subequations}
\begin{align}\label{L1}
L_{1}&=\left(
\begin{array}{cccc} 1 & 0 & 0 & 0\\ -R_{n}^{2} & 0 & 0 & \gamma_{n}^{2}\\
S_{n}R_{n} & 0 & 1 & S_{n}Q_{n}\\-T_{n}R_{n} & \pi_{n}^{2}E &
-T_{n}^{2} & -T_{n}Q_{n}
\end{array}\right)
+\left(\begin{array}{cccc}
 Q_{n}\\ -R_{n}\\ S_{n}\\ -T_{n}
 \end{array}\right) (E-1)^{-1}(R_{n}, S_{n}E, T_{n}, Q_{n}),\\
\label{L2}
L_{2}&=\left(
\begin{array}{cccc} 0 & Q_{n}^{2} & -\gamma_{n}^{2} & 0\\ 0 & -1 & 0 & 0\\
-\pi_{n}^{2}E & S_{n}Q_{n} & R_{n}S_{n} & S_{n}^{2}\\0 & -T_{n}Q_{n}
& -T_{n}R_{n} & -1
\end{array}\right)
+\left(\begin{array}{cccc}
 Q_{n}\\ -R_{n}\\ S_{n}\\ -T_{n}
 \end{array}\right) (E-1)^{-1}(T_{n}E, Q_{n}, R_{n}, S_{n}),
\end{align}
\end{subequations}
in which $\gamma_{n}=\sqrt{1-Q_{n}R_{n}},~
\pi_{n}=\sqrt{1-S_{n}T_{n}}.$
 One can verify that the inverse
operators of $L_{1}$ and $L_{2}$ are
\begin{subequations}
\begin{align}\label{L1-inv}
L_{1}^{-1}=&\left(
\begin{array}{cccc} 1 & 0 & 0 & 0\\ \frac{T_{n-1}R_{n-1}E^{-1}}{\gamma_{n-1}^{2}}
& \frac{T_{n-1}Q_{n-1}E^{-1}}{\gamma_{n-1}^{2}}
 & \frac{T_{n-1}^{2}E^{-1}}{\pi_{n-1}^{2}} & \frac{E^{-1}}{\pi_{n-1}^{2}}\\
-\frac{S_{n}R_{n}}{\gamma_{n}^{2}} &
-\frac{S_{n}Q_{n}}{\gamma_{n}^{2}} & 1 &0\\
\frac{R_{n}^{2}}{\gamma_{n}^{2}}& \frac{1}{\gamma_{n}^{2}} & 0 & 0
\end{array}\right)\notag\\
&-\left(\begin{array}{cccc}
 Q_{n}\\ -T_{n-1}E^{-1}\\ S_{n}\\ -R_{n}
 \end{array}\right) (E-1)^{-1}\biggl(\frac{R_{n}}{\gamma_{n}^{2}}, \frac{Q_{n}}{\gamma_{n}^{2}},
  \frac{T_{n}}{\pi_{n}^{2}}, \frac{S_{n}}{\pi_{n}^{2}}\biggr),
\end{align}
\begin{align}\label{L2-inv}
L_{2}^{-1}=&\left(
\begin{array}{cccc} -\frac{S_{n-1}R_{n-1}E^{-1}}{\gamma_{n-1}^{2}} & -\frac{S_{n-1}Q_{n-1}E^{-1}}{\gamma_{n-1}^{2}}
 & -\frac{E^{-1}}{\pi_{n-1}^{2}} & -\frac{S_{n-1}^{2}E^{-1}}{\pi_{n-1}^{2}}\\
 0 & -1 & 0 & 0\\
-\frac{1}{\gamma_{n}^{2}} & -\frac{Q_{n}^{2}}{\gamma_{n}^{2}} & 0&
0\\\frac{T_{n}R_{n}}{\gamma_{n}^{2}} &
\frac{T_{n}Q_{n}}{\gamma_{n}^{2}} & 0 & -1
\end{array}\right)\notag \\
&-\left(\begin{array}{cccc}
 S_{n-1}E^{-1}\\ -R_{n}\\ Q_{n}\\ -T_{n}
 \end{array}\right) (E-1)^{-1}\biggl(\frac{R_{n}}{\gamma_{n}^{2}}, \frac{Q_{n}}{\gamma_{n}^{2}},
  \frac{T_{n}}{\pi_{n}^{2}}, \frac{S_{n}}{\pi_{n}^{2}}\biggr).
\end{align}
\end{subequations}

To derive isospectral AL hierarchy, we need to take $z_t \equiv 0$ in  \eqref{u-t} and expand $(B_{n}, C_{n}, b_{n}, c_{n})^{T}$ as
\begin{equation}
\label{BCbc+}
\left(
\begin{array}{cccc}
B_{n}\\ C_{n}\\ b_{n}\\ c_{n}
\end{array}\right)=\sum_{j=0}^{m}\left(
\begin{array}{cccc}
B_{n}^{(j)}\\ C_{n}^{(j)}\\ b_{n}^{(j)}\\ c_{n}^{(j)}
\end{array}\right)z^{2(m-j)+1},\qquad m=0,1,2,\cdots.
\end{equation}
Then by setting $(B_{n}^{(0)}, C_{n}^{(0)},
b_{n}^{(0)},c_{n}^{(0)})^{T}=(0, 0, 0, 0)^{T}$,
$A_{0}=-D_{0}=A_{n}|_{u_{n}=0}=-D_{n}|_{u_{n}=0}=\frac{1}{2}z^{2m}$, and comparing the coefficients of
the same powers of $z$ in \eqref{u-t} we get
\begin{subequations}
\label{ut+}
\begin{align}
&\left(
\begin{array}{cccc}
Q_{n}\\ R_{n}\\ S_{n}\\ T_{n}
\end{array}\right)_{t_{m}}
=(1-\delta_{0,m})L_{2}\left(
\begin{array}{cccc}
-B_{n}^{(m)}\\ C_{n}^{(m)}\\ -b_{n}^{(m)}\\ c_{n}^{(m)}
\end{array}\right)+\delta_{0,m}
\left(
\begin{array}{cccc}
Q_{n}\\ -R_{n}\\ S_{n}\\ -T_{n}
\end{array}\right),\label{ut+a}\\
&\left(
\begin{array}{cccc}
-B_{n}^{(j+1)}\\ C_{n}^{(j+1)}\\ -b_{n}^{(j+1)}\\ c_{n}^{(j+1)}
\end{array}\right)=-L_{1}^{-1}L_{2}\left(
\begin{array}{cccc}
-B_{n}^{(j)}\\ C_{n}^{(j)}\\ -b_{n}^{(j)}\\ c_{n}^{(j)}
\end{array}\right), \qquad j=1,2,\cdots, m-1,\\
&\left(
\begin{array}{cccc}
-B_{n}^{(1)}\\ C_{n}^{(1)}\\ -b_{n}^{(1)}\\ c_{n}^{(1)}
\end{array}\right)=-L_{1}^{-1}\left(
\begin{array}{cccc}
Q_{n}\\ -R_{n}\\ S_{n}\\ -T_{n}
\end{array}\right),\label{ut+c}
\end{align}
\end{subequations}
where the subindex $m$ for $t$ indicates the order of the expansion
\eqref{BCbc+} as well as the order of member in isospectral
hierarchy. This further yields a isospectral hierarchy
\begin{equation}
\label{hie-iso+}
u_{n, t_{m}}=K^{(m)}=L^{m}K^{(0)}, \qquad m=0, 1, 2, \cdots,
\end{equation}
where
\begin{equation}
\label{K0}
 K^{(0)}=(Q_{n}, -R_{n}, S_{n}, -T_{n})^{T},
\end{equation}
 and the
recursion operator $L$ is defined by
\begin{align}
\label{L}
L=-L_{2}L_{1}^{-1}= &\left(
\begin{array}{cccc} -S_{n}R_{n} & -S_{n}Q_{n} & \gamma_{n}^{2} & -Q_{n}^{2}E^{-1} \\
0 & 0 & 0 & E^{-1}\\ \pi_{n}^{2}E & -S_{n}^{2} & -S_{n}R_{n} &
-S_{n}Q_{n}E^{-1}\\
\frac{\pi_{n}^{2}R_{n}^{2}}{\gamma_{n}^{2}} &
\frac{1-Q_{n}R_{n}S_{n}T_{n}}{\gamma_{n}^{2}} & T_{n}R_{n} &
T_{n}Q_{n}E^{-1}
\end{array}\right)\notag\\
&-\left(\begin{array}{cccc}
 Q_{n}\\ -R_{n}\\ S_{n}\\ -T_{n}
 \end{array}\right) (E-1)^{-1}(T_{n}E, S_{n}, R_{n}, Q_{n}E^{-1})\notag\\
&
 -\left(\begin{array}{cccc}
 S_{n}\gamma_{n}^{2}\\ -T_{n-1}\gamma_{n}^{2}\\ Q_{n+1}\pi_{n}^{2}E\\
 -R_{n}\pi_{n}^{2}
 \end{array}\right) (E-1)^{-1}\biggl(\frac{R_{n}}{\gamma_{n}^{2}}, \frac{Q_{n}}{\gamma_{n}^{2}},
  \frac{T_{n}}{\pi_{n}^{2}}, \frac{S_{n}}{\pi_{n}^{2}}\biggr).
\end{align}

If we expand $(B_{n}, C_{n}, b_{n}, c_{n})^{T}$ in another direction, i.e.,
\begin{equation}\label{2.22}
\left(
\begin{array}{cccc}
B_{n}\\ C_{n}\\ b_{n}\\ c_{n}
\end{array}\right)=\sum_{j=m}^{0}\left(
\begin{array}{cccc}
{B}_{n}^{(j)}\\ {C}_{n}^{(j)}\\ {b}_{n}^{(j)}\\
{c}_{n}^{(j)}
\end{array}\right)z^{2(m-j)-1},\qquad m=0,-1,-2,\cdots,
\end{equation}
and take $({B}_{n}^{(0)}, {C}_{n}^{(0)}, {b}_{n}^{(0)},
{c}_{n}^{(0)})^{T}=(0, 0, 0, 0)^{T}$,
$A_{0}=-D_{0}=A_{n}|_{u_{n}=0}=-D_{n}|_{u_{n}=0}=\frac{1}{2}z^{2m}$, we can have relation
\begin{subequations}
\label{ut-}
\begin{align}
&\left(
\begin{array}{cccc}
Q_{n}\\ R_{n}\\ S_{n}\\ T_{n}
\end{array}\right)_{t_{m}}
=(1-\delta_{0,m})L_{1}\left(
\begin{array}{cccc}
-B_{n}^{(m)}\\ C_{n}^{(m)}\\ -b_{n}^{(m)}\\ c_{n}^{(m)}
\end{array}\right)+\delta_{0,m}
\left(
\begin{array}{cccc}
Q_{n}\\ -R_{n}\\ S_{n}\\ -T_{n}
\end{array}\right),\\
&\left(
\begin{array}{cccc}
-B_{n}^{(j)}\\ C_{n}^{(j)}\\ -b_{n}^{(j)}\\ c_{n}^{(j)}
\end{array}\right)=-L_{2}^{-1}L_{1}\left(
\begin{array}{cccc}
-B_{n}^{(j+1)}\\ C_{n}^{(j+1)}\\ -b_{n}^{(j+1)}\\ c_{n+1}^{(j)}
\end{array}\right), \qquad j=-2,-3,\cdots, m,\\
&\left(
\begin{array}{cccc}
-B_{n}^{(-1)}\\ C_{n}^{(-1)}\\ -b_{n}^{(-1)}\\ c_{n}^{(-1)}
\end{array}\right)=-L_{2}^{-1}\left(
\begin{array}{cccc}
Q_{n}\\ -R_{n}\\ S_{n}\\ -T_{n}
\end{array}\right),
\end{align}
\end{subequations}
and further get a negative order isospectral hierarchy
\begin{equation}
\label{hie-iso-} u_{n, t_{m}}=K^{(m)}=L^{m}K^{(0)}, \qquad m=0, -1,
-2, \cdots,
\end{equation}
where $K^{(0)}$ and $L$ are given by \eqref{K0} and \eqref{L} respectively.

Obviously, \eqref{hie-iso+} and \eqref{hie-iso-} can be jointed
together and written as a uniformed isospectral AL hierarchy
\begin{equation}
\label{hie-iso} u_{n, t_m}=K^{(m)}=L^{m}K^{(0)}, \qquad m\in
\mathbb{Z}.
\end{equation}
In Appendix \ref{A-1} we will list out the first few isospectral equations and their related Lax pairs.

\subsection{Non-isospectral hierarchy}


For the non-isospectral case, we suppose the time evolution of
spectral parameter $z$ follows $z_{t_m}=\frac{1}{2}z^{2m+1}$ for any
given $m\in \mathbb{Z}$. We can first expand $(B_{n}, C_{n}, b_{n},
c_{n})^{T}$ as \eqref{BCbc+} and still take $(B_{n}^{(0)},
C_{n}^{(0)}, b_{n}^{(0)}, c_{n}^{(0)})^{T}=(0, 0, 0, 0)^{T}$ but
$A_{0}=-D_{0}=0$. Then we can get a recursion relation which is
similar to \eqref{ut+} but in this non-isospectral case \eqref{ut+a}
and \eqref{ut+c} are replaced by
\begin{align*}
&\left(
\begin{array}{cccc}
Q_{n}\\ R_{n}\\ S_{n}\\ T_{n}
\end{array}\right)_{t_{m}}
=(1-\delta_{0,m})L_{2}\left(
\begin{array}{cccc}
-B_{n}^{(m)}\\ C_{n}^{(m)}\\ -b_{n}^{(m)}\\ c_{n}^{(m)}
\end{array}\right)+\delta_{0,m}
\left(
\begin{array}{cccc}
(2n+\frac{1}{2})Q_{n}\\ -(2n+\frac{1}{2})R_{n}\\ (2n+\frac{3}{2})S_{n}\\
-(2n+\frac{3}{2})T_{n}
\end{array}\right),\\
&\left(
\begin{array}{cccc}
-B_{n}^{(1)}\\ C_{n}^{(1)}\\ -b_{n}^{(1)}\\ c_{n}^{(1)}
\end{array}\right)=-L_{1}^{-1}\left(
\begin{array}{cccc}
(2n+\frac{1}{2})Q_{n}\\ -(2n+\frac{1}{2})R_{n}\\
(2n+\frac{3}{2})S_{n}\\ -(2n+\frac{3}{2})T_{n}
\end{array}\right).
\end{align*}
It then follows that a positive order non-isospectral hierarchy is
\begin{equation}
\label{2.34}
u_{n, t_m}=\sigma^{(m)}=L^{m}\sigma^{(0)}, \qquad m=0, 1, 2, \cdots,
\end{equation}
where
\begin{equation}
\label{sigma0}
 \sigma^{(0)}=\left(
\begin{array}{cccc}
(2n+\frac{1}{2})Q_{n}\\ -(2n+\frac{1}{2})R_{n}\\
(2n+\frac{3}{2})S_{n}\\ -(2n+\frac{3}{2})T_{n}
\end{array}\right),
\end{equation}
and $L$ is the recursion operator given by \eqref{L}.
After a discussion for a negative order expansion of $(B_{n}, C_{n}, b_{n}, c_{n})^{T}$
and deriving a negative order non-isospectral hierarchy, one can finally reach to
a uniformed non-isospectral AL hierarchy:
\begin{equation}
\label{hie-non}
u_{n, t_m}=\sigma^{(m)}=L^{m}\sigma^{(0)}, \qquad m\in \mathbb{Z}.
\end{equation}

In Appendix \ref{A-1} we will also list out the first few
non-isospectral equations and their related Lax pairs.

\subsection{Zero-curvature representations}

We have derived isospectral hierarchy \eqref{hie-iso} and non-isospectral hierarchy \eqref{hie-non}.
Suppose that their Lax pairs are respectively
\begin{equation}
\Phi_{n+1}=U_n\Phi_{n},~~~\Phi_{n,t_m}=G_n^{(m)}\Phi_n,
\label{Laxp-iso}
\end{equation}
and
\begin{equation}
\Phi_{n+1}=U_n\Phi_{n},~~~\Phi_{n,t_m}=W_n^{(m)}\Phi_n,
\label{Laxp-non}
\end{equation}
where $m\in\mathbb{Z}$ and $U_n$ is defined in \eqref{4p-new}.
Then in isospectral case the zero-curvature equation which is related to \eqref{Laxp-iso} is
\begin{equation}
U_{n,t_m}=(EG^{(m)}_{n})U_{n}-U_{n}G^{(m)}_{n}, \label{zce-iso}
\end{equation}
and in non-isospectral case
\begin{equation}
U_{n,t_m}=(EW^{(m)}_{n})U_{n}-U_{n}W^{(m)}_{n}. \label{zce-non}
\end{equation}

Noticing the definition \eqref{def-gat} for a Gateaux derivative, $U_{n,t_m}$ can be rewritten as
\begin{equation}
U_{n,t_m}=U_n'[u_{n,t_m}]+U_{n,z}\cdot z_{t_m},
\end{equation}
by means of which we have
\begin{proposition}\label{p1}
The isospectral flows $\{K^{(m)}\}$ and non-isospectral flows
$\{\sigma^{(m)}\}$ admit  zero curvature representations
\begin{align}
& U_{n}^{\prime}[K^{(m)}]=(EG^{(m)}_{n})U_{n}-U_{n}G^{(m)}_{n},\label{zcr-iso}\\
&
U_{n}^{\prime}[\sigma^{(m)}]=(EW^{(m)}_{n})U_{n}-U_{n}W^{(m)}_{n}-\frac{1}{2}z^{2m+1} U_{n,z},\label{zcr-non}
\end{align}
where $m\in \mathbb{Z}$,
$K^{(m)}$ and $\sigma^{(m)} \in \mathcal{V}_{4}$, $G^{(m)}_{n}$ and
$ W^{(m)}_{n} \in \mathcal{Q}_{2}(z)$ and satisfy
\begin{equation}\label{2.44}
G^{(m)}_{n}|_{u_{n}=0}=\frac{z^{2m}}{2}\left(
\begin{array}{cc} 1 & 0\\ 0& -1
\end{array}\right),~~
W^{(m)}_{n}|_{u_{n}=0}=z^{2m}\left(
\begin{array}{cc} n & 0\\ 0& -n
\end{array}\right).
\end{equation}
\end{proposition}

Besides, noting that $U_n'$ is an injective homomorphism when  $Q_nR_n+S_nT_n\neq 0$, we have (cf.\cite{ZDJ-02-JPA,ZDJ-06-PLA})
\begin{lemma}
The matrix equation
\begin{equation}
U_{n}^{\prime}[X_n]=(EV_{n})U_{n}-U_{n}V_{n}, ~~X_n\in
\mathcal{V}_{4},~ V_{n}\in \mathcal{Q}_{2}(z) ~\mathrm{and}
~V_{n}|_{u_{n}=0}=0
\end{equation}
has only zero solutions $X_n=0$ and $V_{n}=0$.
\end{lemma}

This lemma and the zero-curvature representations \eqref{zcr-iso} and \eqref{zcr-non}
will play important roles in constructing symmetries and determining their algebraic structures.

\subsection{Recursion operator and hereditary and strong symmetry}

The flows $\{K^{(m)}\}, \{\sigma^{(m)}\}$ and their recursion relation can also be derived from the following way.

Firstly, we start from the matrix equations
\begin{align}
& U_{n}^{\prime}[X_n]=(EG_{n})U_{n}-U_{n}G_{n},\label{zcr-iso-0}\\
& U_{n}^{\prime}[Z_n]=(EW_{n})U_{n}-U_{n}W_{n}-U_{n,z}\cdot
\frac{z}{2},\label{zcr-non-0}
\end{align}
where the unknowns are $X_n, Z_n\in \mathcal{V}_{4}$ and
$G_{n},W_n\in\mathcal{Q}_{2}(z)$. Obviously, both of the two
equations admit non-zero solution pairs $\{X_n, G_n\}$ and $\{Z_n,
W_n\}$ (see \eqref{zcr-iso} and \eqref{zcr-non} for $m=0$).
Secondly, for equation
\begin{equation}
\label{zce-rec} U_{n}'[X_n-z^{\alpha}Y_n]=(EV_{n})U_{n}-U_{n}V_{n}
\end{equation}
where $X_n, Y_n\in \mathcal{V}_{4}$ and
$V_{n}\in\mathcal{Q}_{2}(z)$, one can find that when $\alpha=2$, for
any given $Y_n\neq 0\in \mathcal{V}_{4}$, there exist unique
solutions $X_n\in \mathcal{V}_{4}$, $V_{n}\in\mathcal{Q}_{2}(z)$
where $V_{n}|_{u_n=0}=0$. Thus $X_n$ and $Y_n$ are connected by a
map:\footnote{If $\alpha=-2$ then the map is $X=L^{-1}Y$.}
\begin{equation}
\label{map:L} L:~~ X_n=LY_n.
\end{equation}
The map, $L$, is nothing but the recursion operator \eqref{L}. Then,
thirdly, let $Y_n=K^{(0)}$ or $\sigma^{(0)}$ in \eqref{zce-rec}, one
can get $K^{(1)}$ or $\sigma^{(1)}$ by taking $\alpha=2$ and
$K^{(-1)}$ or $\sigma^{(-1)}$ by taking $\alpha=-2$. Repeating the
procedure one can generate all the isospectral flows $\{K^{(m)}\}$
and non-isospectral flows $\{\sigma^{(m)}\}$.

Besides, with the recursion operator $L$ in hand, \eqref{zce-rec} becomes
\begin{equation}
U_{n}'[X-z^{2}LX]=(EV_{n})U_{n}-U_{n}V_{n}.
\end{equation}
Then, in the light of Theorem 1 in Ref.\cite{ZDJ-02-JPA} (also see Appendix \ref{A-2}),
we immediately  have
\begin{proposition}
The recursion operator $L$ is hereditary and a strong symmetry\footnote{
For the definitions of a hereditary operator and a strong symmetry operator,
one can refer to \cite{Fuchssteiner-1981}.} for the  isospectral AL hierarchy (\ref{hie-iso}).
\end{proposition}

\section{Symmetries and Lie algebras}

In this section, we construct two types of symmetries for the
isospectral and non-isospectral AL hierarchies, respectively.
To do that, let us first look at the algebraic structure of flows $\{K^{(m)}\}$ and
$\{\sigma^{(m)}\}$.
We note that the proofs for this section are quite similar to those in \cite{ZDJ-06-PLA} and here we skip them.

\subsection{Algebra of flows}

Making use of the identity\cite{Olver-book}
\begin{equation}
U_n'[\K  f, g \J   ]=(U_n'[f])'[g]-(U_n'[g])'[f],~~~ \forall f,g \in \mathcal{V}_{4},
\label{identity-1}
\end{equation}
from the zero-curvature representations \eqref{zcr-iso} and \eqref{zcr-non}
we can derive the following relations:
\begin{lemma}\label{lem-1}
The isospectral flows $\{K^{(m)}\}$ and non-isospectral flows $\{\sigma^{(m)}\}$ satisfy
\begin{subequations}
\label{alg-flow}
\begin{align}
&U_{n}^{\prime}[\K  K^{(m)},
K^{(s)}\J  ]=(E<G_{n}^{(m)}, G_{n}^{(s)}>)U_{n}-U_{n}<G_{n}^{(m)}, G_{n}^{(s)}>,\label{45a}\\
&U_{n}^{\prime}[\K  K^{(m)},
\sigma^{(s)}\J  ]=(E<G_{n}^{(m)}, W_{n}^{(s)}>)U_{n}-U_{n}<G_{n}^{(m)}, W_{n}^{(s)}>,\label{45b}\\
&U_{n}^{\prime}[\K  \sigma^{(m)},
\sigma^{(s)}\J  ]=(E<W_{n}^{(m)},
W_{n}^{(s)}>)U_{n}-U_{n}<W_{n}^{(m)},
W_{n}^{(s)}>-\frac{1}{2}(m-s)U_{n, z}z^{2(m+s)+1},\label{45c}
\end{align}
\end{subequations}
where
\begin{subequations}\label{46}
\begin{align}
& <G_{n}^{(m)},
G_{n}^{(s)}>=G_{n}^{(m)\prime}[K^{(s)}]-G_{n}^{(s)\prime}[K^{(m)}]+[G_{n}^{(m)},
G_{n}^{(s)}],\label{46a}\\
& <G_{n}^{(m)},
W_{n}^{(s)}>=G_{n}^{(m)\prime}[\sigma^{(s)}]-W_{n}^{(s)\prime}[K^{(m)}]+[G_{n}^{(m)},
W_{n}^{(s)}]+\frac{1}{2}G_{n, z}^{(m)}z^{2s+1},\label{46b}\\
&<W_{n}^{(m)},
W_{n}^{(s)}>=W_{n}^{(m)\prime}[\sigma^{(s)}]-W_{n}^{(s)\prime}[\sigma^{(m)}]+[W_{n}^{(m)},
W_{n}^{(s)}]+\frac{1}{2}W_{n, z}^{(m)}z^{2s+1}-\frac{1}{2}W_{n,
z}^{(s)}z^{2m+1}.\label{46c}
\end{align}
\end{subequations}
\end{lemma}

This lemma can be proved via a similar procedure as in \cite{ZDJ-06-PLA}
and here we skip the proof.
Based on the lemma we then come up with a algebra for the flows  $\{K^{(m)}\}$ and   $\{\sigma^{(m)}\}$ (cf.\cite{ZDJ-06-PLA}).
\begin{lemma}
\label{lem-2}
The isospectral and non-isospectral flows, $\{K^{(m)}\}$ and
$\{\sigma^{(m)}\}$, compose an
infinite-dimensional Lie algebra $\mathcal{F}$ through the Lie
product $\K \cdot, \cdot \J  $ and possess the following  relations
\begin{subequations}
\label{47}
\begin{align}
\K  K^{(m)}, K^{(s)}\J   & =0, \\
\K  K^{(m)},\sigma^{(s)}\J   &=mK^{(m+s)}, \\
\K  \sigma^{(m)},\sigma^{(s)}\J   &=(m-s)\sigma^{(m+s)}.
\end{align}
\end{subequations}
\end{lemma}

\subsection{Symmetries for the isospectral and non-isospectral AL hierarchies}

\begin{theorem}
\label{T-4.1}
Any given member $u_{n, t_m}=K^{(m)}$ in the isospectral four-potential AL hierarchy \eqref{hie-iso}
possesses the following two sets of symmetries, i.e.,
\begin{equation}
\label{sym-iso} \{K^{(s)}\} ~~~\mathrm{and} ~~~ \{\tau^{(m,
s)}=m t_m K^{(m+s)}+\sigma^{(s)}\}, ~~~ s\in \mathbb{Z}.
\end{equation}
These symmetries  form a centerless Kac-Moody-Virasoro (KMV) algebra
$\mathcal{S}$ with the following structure
\begin{subequations}
\label{alg-sym-iso}
\begin{align}
\K  K^{(l)}, K^{(s)}\J   &=0,\\
\K  K^{(l)}, \tau^{(m,s)}\J   &=lK^{(l+s)},\\
\K  \tau^{(m, s)}, \tau^{(m, l)} \J   &=(s-l)\tau^{(m, s+l)}.
\end{align}
\end{subequations}
\end{theorem}

Obviously, the Lie algebras $\mathcal{F}$ and $\mathcal{S}$ are
respectively generated by the following elements
\begin{align}
&\{\sigma^{(1)}~~ (\mathrm{or} ~\sigma^{(-1)}),~~~~ \sigma^{(2)},~~~~
\sigma^{(-2)},~~~~ K^{(1)}~~
(\mathrm{or} ~K^{(-1)})\};\label{51}\\
&\{\tau^{(m, 1)}~~ (\mathrm{or} ~\tau^{(m, -1)}),~~~~\tau^{(m, 2)},~~~~
\tau^{(m, -2)}, ~~~~K^{(1)}~~ (\mathrm{or}~ K^{(-1)})\}.\label{52}
\end{align}

\begin{theorem}
\label{T-4.2} Any given member $u_{n, t_m}=\sigma^{(m)}$ in the
non-isospectral four-potential AL hierarchy \eqref{hie-non} has two
sets of symmetries, i.e.,
\begin{subequations}
\label{sym-non}
\begin{align}
&\eta^{(m, s)}=\sum^{s}_{j=0}C_{s}^{j}(m t_m)^{s-j}\sigma^{(m-jm)}~~~~(s=0, 1, 2, \cdots),\label{56}\\
&\gamma^{(m, s)}
=\sum^{s}_{j=0}C_{s}^{j}(m t_m)^{s-j}K^{(-jm)}~~~~(s=0, 1, 2,
\cdots),\label{57}
\end{align}
\end{subequations}
which we call  $\eta$-symmetries and $\gamma$-symmetries,
respectively. Here $C^{j}_{s}=\frac{s!}{j!(s-j)!} $. These
symmetries form a centerless  KMV algebra ${\mathcal{H}}$ with the
following structure
\begin{subequations}
\label{alg-sym-non}
\begin{align}
\K  \eta^{(m, s)}, \eta^{(m, l)}\J   & =(l-s)m \eta^{(m, s+l-1)},\label{58}\\
\K  \gamma^{(m, s)}, \gamma^{(m, l)}\J   & =0,\label{59}\\
\K  \eta^{(m, s)}, \gamma^{(m, l)}\J   &
=lm\gamma^{(m,s+l-1)}.\label{60}
\end{align}
\end{subequations}
The algebra can be generated by
\begin{equation}
\label{61}
\{\eta^{(m, 0)},~~~~ \eta^{(m, 3)}, ~~~~\gamma^{(m, 1)}\}.
\end{equation}
\end{theorem}

\subsection{Relations between the recursion operator and flows}

\begin{theorem}
The isospectral flows $\{K^{(m)}\}$, non-isospectral flows $\{\sigma^{(m)}\}$ and recursion operator $L$
satisfy the relations
\begin{align}
&L^{\prime}[K^{(m)}]-[K^{(m)\prime}, L]=0,\label{53}\\
&L^{\prime}[\sigma^{(m)}]-[\sigma^{(m)\prime},
L]-L^{m+1}=0,\label{54}
\end{align}
where $m\in\mathbb{Z}$ and $[A,B]=AB-BA$.
\end{theorem}

Utilizing the relations one can also derive symmetries \eqref{sym-iso} and their algebra structure \eqref{alg-sym-iso}
by means of inductive approach, as in \cite{Tian-book-90}.
In that way the non-isospectral flow $\sigma^{(0)}$ will play the role of a master symmetry \cite{Fuchssteiner-1983}.

\section{Reduction to two-potential case}

It is possible to reduce isospectral, non-isospectral AL hierarchies and their symmetries
from four-potential case to two-potential case.

\subsection{Reduction of $(S_n,T_n)=(0,0)$}

\subsubsection{Spectral problem}

Let us start from the four-potential AL-spectral problem \eqref{4p-new}.
Taking $(S_n,T_n)=(0,0)$ in \eqref{4p-new} yields
\begin{equation}
\Phi_{n+1}=\overline{U}_n\Phi_{n},~~~
\overline{U}_n=\left(
\begin{array}{cc} z^{2} & zQ_{n}\\
z^{-1}R_{n}             & z^{-2}
\end{array}\right),~~~\Phi_{n}=\left(
\begin{array}{cc} \phi_{1,n}\\
\phi_{2,n}
\end{array}\right),
\label{2p-AL-new}
\end{equation}
but this is not the canonical form \eqref{2p-AL} yet.
Next we introduce a gauge transformation
\begin{equation}
\label{gauge}
\Psi_{n}=\left(
\begin{array}{cc} z^{-1} & 0\\
0             & 1
\end{array}\right)\Phi_n,~~~\Psi_{n}=\left(
\begin{array}{cc} \psi_{1,n}\\
\psi_{2,n}
\end{array}\right),
\end{equation}
under which \eqref{2p-AL-new} becomes
\begin{equation}
\Psi_{n+1}=M_n\Psi_{n},~~~
M_n=\left(
\begin{array}{cc} z^{2} & Q_{n}\\
R_{n}             & z^{-2}
\end{array}\right),
\label{2p-AL-can}
\end{equation}
which is the canonical form \eqref{2p-AL} in the light of $z^2=\lambda$.

\subsubsection{Closeness discussion}

For convenience we introduce some notations. Let
\begin{equation}
\overline{u}_{n}=\left(
\begin{array}{c} Q_{n}\\R_{n}
\end{array}\right),~~
{\overline{K}}^{(m)}=\left(
\begin{array}{c} K_1^{(m)}\\K_2^{(m)}
\end{array}\right)_{(S_n,T_n)=(0,0)},~~
\overline{\sigma}^{(m)}=\left(
\begin{array}{c} \sigma_1^{(m)}\\ \sigma_2^{(m)}
\end{array}\right)_{(S_n,T_n)=(0,0)},~~
\end{equation}
where $K_j^{(m)}$ and $\sigma_j^{(m)}$ are the $j$-th elements of $K^{(m)}$ and $\sigma^{(m)}$.

Noting that the reduction $(S_n,T_n)=(0,0)$ yields
$u_n|_{(S_n,T_n)=(0,0)}=(\overline{u}_{n}^T,0,0)^T$, it is necessary
to discuss whether the reduction is closed, i.e., whether
\begin{equation}
{K}^{(m)}\bigr|_{(S_n,T_n)=(0,0)}=\left(
\begin{array}{c} \overline{K}^{(m)}\\0\\0
\end{array}\right),~~
\sigma^{(m)}\bigr|_{(S_n,T_n)=(0,0)}=\left(
\begin{array}{c} \overline{\sigma}^{(m)}\\0\\0
\end{array}\right).
\end{equation}
Next we will see  this is true for $m=2h, h\in \mathbb{Z}$.
\begin{theorem}
\label{T-reduce}
Under the reduction $(S_n,T_n)=(0,0)$, those even order members in the four-potential AL hierarchies
\eqref{hie-iso} and \eqref{hie-non} reduce to two-potential isospectral AL-hierarchies
\begin{equation}
\label{hie-iso-2p}
\overline{u}_{n, t_{2h}}=\overline{K}^{(2h)}=\overline{L}^{h}\overline{K}^{(0)}, \qquad h\in \mathbb{Z}
\end{equation}
and
\begin{equation}
\label{hie-non-2p}
\overline{u}_{n, t_{2h}}=\overline{\sigma}^{(2h)}=\overline{L}^{h}\overline{\sigma}^{(0)}, \qquad h\in \mathbb{Z},
\end{equation}
where
\begin{equation}
\overline{K}^{(0)}=(Q_{n},
-R_{n})^T,~~~\overline{\sigma}^{(0)}=(2n+\frac{1}{2})(Q_{n},
-R_{n})^T, \label{K0-sigma0-2p}
\end{equation}
and the recursion operator $\overline{L}$ is
\begin{equation}
\overline{L}=\!\!\biggl(
\begin{array}{cc} E& 0\\0 & E^{-1}
\end{array}\biggr)\!+\left(
\begin{array}{cc} -Q_{n}E\\R_{n}
\end{array}\right)\!(E-1)^{-1}(R_{n}E, Q_{n}E^{-1})+\gamma_{n}^{2}\left(
\begin{array}{cc} -Q_{n+1} E\\R_{n-1}
\end{array}\right)\!(E-1)^{-1}\biggl(\frac{R_{n}}{\gamma_{n}^{2}},
\frac{Q_{n}}{\gamma_{n}^{2}}\biggr).
\end{equation}
\end{theorem}

\begin{proof}
First, look at \eqref{K0} and \eqref{sigma0}. When $m=0$, the reduction $(S_n,T_n)=(0,0)$ is obviously closed and the resulting flows
$\overline{K}^{(0)}$ and $\overline{\sigma}^{(0)}$
are as given in \eqref{K0-sigma0-2p}. Meanwhile, by direct calculation
we find
\begin{align*}
L^2\bigr|_{(S_n,T_n)=(0,0)}= \left(
\begin{array}{cc} \overline{L}& \mathbf{0}\\
                   \mathbf{0}  & H
\end{array}\right),
\end{align*}
where
\begin{align*}
H=\left(
\begin{array}{cc} \gamma_{n+1}^{2}E & -Q_{n+1}^{2} \\
                   R_{n}^{2}  &
                   (1+Q_{n}R_{n})E^{-1}
\end{array}\right)-2\left(
\begin{array}{cc} Q_{n+1}E \\
                   -R_{n}\end{array}\right)(E-1)^{-1}(R_{n}, Q_{n}E^{-1}).
\end{align*}
That means the closeness is valid for those even order members of four-potential AL hierarchies
\eqref{hie-iso} and \eqref{hie-non} when we take $(S_n,T_n)=(0,0)$.
Thus the proof is completed.
\end{proof}

\subsubsection{Flows under gauge transformation}

With the closeness in hand, let us make a comparison  for the present results and those in \cite{ZDJ-06-PLA}.
Isospectral flows and recursion operator are exactly same but non-isospectral flows are different.
The basic non-isospectral flow  $\widehat{\sigma}^{(0)}$ given in \cite{ZDJ-06-PLA} is
\begin{equation}
\widehat{\sigma}^{(0)}=(2n+1)(Q_{n}, -R_{n})^T,
\label{sigma0-2p-PLA}
\end{equation}
and the difference from $\overline{\sigma}^{(0)}$ is
\begin{equation}
\overline{\sigma}^{(0)}-\widehat{\sigma}^{(0)}=-\frac{1}{2}(Q_{n},
-R_{n})^T.
\label{sigma0-2p-dif}
\end{equation}
To understand the difference on non-isospectral flows, we go back to the gauge transformation \eqref{gauge}.

Obviously, the Lax pair of $\overline{u}_{n, t_{0}}=\overline{\sigma}^{(0)}$
is composed by \eqref{2p-AL-new} and
\begin{equation}
\label{2p-new-time}
\Phi_{n,t_0}=\overline{V}^{(0)}_{n}\Phi_{n}, \qquad \overline{V}^{(0)}_n={V}^{(0)}_n |_{(S_n,T_n)=(0,0)},
\end{equation}
of which the zero-curvature equation reads
\begin{equation}
\label{2p-new-zce}
\overline{U}_{n,
t_0}=(E\overline{V}^{(0)}_{n})\overline{U}_{n}-\overline{U}_{n}\overline{V}^{(0)}_{n}.
\end{equation}
Under the gauge transformation \eqref{gauge}, \eqref{2p-new-time} turns out to be
\begin{equation}
\Psi_{n, t_0}=\widehat{V}^{(0)}_{n}\Psi_{n}
\end{equation}
where
\begin{equation}
\widehat{V}^{(0)}_{n}=\widehat{N}^{(0)}_{n}-
\left(\begin{array}{cc} z^{-1}z_{t}& 0\\
                   0  & 0
\end{array}\right),
~~z_t=\frac{z}{2},~~
\widehat{N}^{(0)}_{n}=\left(
\begin{array}{cc} z^{-1}& 0\\
                   0  & 1
\end{array}\right)
\overline{V}^{(0)}_{n}\left(
\begin{array}{cc} z& 0\\
                   0  & 1
\end{array}\right).
\end{equation}
Meanwhile the
zero-curvature equation \eqref{2p-new-zce} is transformed to
\begin{equation}
M_{n, t_0}-(E\widehat{V}^{(0)}_{n})M_{n}+M_{n}\widehat{V}^{(0)}_{n}=0,
\end{equation}
i.e.,
\begin{equation}
M_{n, t_0}-(E\widehat{N}^{(0)}_{n})M_{n}+M_{n}\widehat{N}^{(0)}_{n}
=\frac{1}{2}\left(
\begin{array}{cc} 0& -Q_{n}\\
            R_{n}  & 0
\end{array}\right).
\label{zce-gu-tr}
\end{equation}
Compared with \cite{ZDJ-06-PLA} we find that
the l.h.s. of \eqref{zce-gu-tr} equating to zero yields the non-isospectral equation
$\overline{u}_{n,t_0}=\widehat{\sigma}^{(0)}$ while
the r.h.s. just gives the difference \eqref{sigma0-2p-dif}.

Now we can see that it is just the time-dependent multiplier
$\left(
\begin{smallmatrix} z^{-1}& 0\\
                   0  & 1
\end{smallmatrix}\right)$
in the transformation \eqref{gauge} contributes the extra term \eqref{sigma0-2p-dif} for the
non-isospectral flow $\overline{\sigma}^{(0)}$,
but in isospectral case this multiplier adds nothing new since in this turn $z$ is independent of time.
However, the extra term in \eqref{sigma0-2p-dif} is nothing but $\frac{1}{2}\overline{K}^{(0)}$,
which means the obtained non-isospectral flow $\overline{\sigma}^{(2h)}$
is only a summation of $\widehat{\sigma}^{(h)}$ and $-\frac{1}{2}\overline{K}^{(2h)}$,
where $\widehat{\sigma}^{(h)}=\overline{L}^{h}\widehat{\sigma}^{(0)}$
is the non-isospectral flow derived in \cite{ZDJ-06-PLA}.

\subsection{Symmetries}

The closeness also guarantees reduction of symmetries and their algebraic structures.
From Theorem \ref{T-4.1} and \ref{T-4.2} one can directly have
\begin{theorem}
\label{T-5.2}
Any given member $\overline{u}_{n, t_{2h}}=\overline{K}^{(2h)}$ in the isospectral two-potential AL hierarchy \eqref{hie-iso-2p}
can have two sets of symmetries
\begin{equation}
\label{sym-iso-2p}
\{\overline{K}^{(2s)}\} ~~~\mathrm{and} ~~~
\{\overline{\tau}^{(2h,
2s)}=2ht_{2h}\overline{K}^{(2h+2s)}+\overline{\sigma}^{(2s)}\}, ~~~
s\in \mathbb{Z},
\end{equation}
which  form a centerless KMV algebra $\overline{\mathcal{S}}$ with
structure
\begin{subequations}
\label{alg-sym-iso-2p}
\begin{align}
\K  \overline{K}^{(2l)}, \overline{K}^{(2s)}\J   &=0,\\
\K  \overline{K}^{(2l)}, \overline{\tau}^{(2h,2s)}\J   &=2l\overline{K}^{(2l+2s)},\\
\K  \overline{\tau}^{(2h, 2s)}, \overline{\tau}^{(2h, 2l)}
\J   &=2(s-l)\overline{\tau}^{(2h, 2s+2l)},
\end{align}
\end{subequations}
and generators
\begin{equation}
\{\overline{\tau}^{(2h, 2)}~~ (\mathrm{or} ~\overline{\tau}^{(2h,
-2)}),~~~~\overline{\tau}^{(2h, 4)},~~~~ \overline{\tau}^{(2h,
-4)}, ~~~~\overline{K}^{(2)}~~ (\mathrm{or}~ \overline{K}^{(-2)})\}.
\end{equation}
\end{theorem}

\begin{theorem}
\label{T-5.3}
Any given member $\overline{u}_{n, t_{2h}}=\overline{\sigma}^{(2h)}$ in the non-isospectral two-potential AL hierarchy \eqref{hie-non-2p}
can have two sets of symmetries
\begin{subequations}
\label{sym-non-2p}
\begin{align}
&\overline{\eta}^{(2h, s)}=\sum^{s}_{j=0}C_{s}^{j}(2ht_{2h})^{s-j}\overline{\sigma}^{(2h-2jh)}~~~~(s=0, 1, 2, \cdots),\label{56}\\
&\overline{\gamma}^{(2h, s)}
=\sum^{s}_{j=0}C_{s}^{j}(2ht_{2h})^{s-j}\overline{K}^{(-2jh)}~~~~(s=0,
1, 2, \cdots),\label{57}
\end{align}
\end{subequations}
which  form a centerless KMV algebra $\overline{\mathcal{H}}$ with
structure
\begin{subequations}
\label{alg-sym-non}
\begin{align}
\K  \overline{\eta}^{(2h, s)}, \overline{\eta}^{(2h, l)}\J   & =2(l-s)h \overline{\eta}^{(2h, s+l-1)},\label{58}\\
\K  \overline{\gamma}^{(2h, s)}, \overline{\gamma}^{(2h, l)}\J   & =0,\label{59}\\
\K  \overline{\eta}^{(2h, s)}, \overline{\gamma}^{(2h,
l)}\J   & =2lh\overline{\gamma}^{(2h,s+l-1)},\label{60}
\end{align}
\end{subequations}
and generators
\begin{equation}
\{\overline{\eta}^{(2h, 0)},~~~~ \overline{\eta}^{(2h, 3)},
~~~~\overline{\gamma}^{(2h, 1)}\}.
\end{equation}
\end{theorem}

\section{Conclusions}

We have derived isospectral and non-isospectral four-potential AL
hierarchies and their recursion operator. Both hierarchies have been
shown to have infinitely many symmetries and Lie algebras which
belong to centerless KMV algebras, respectively. These structures
are derived by means of zero-curvature representations of flows.
What is special is that these structures cover both positive and
negative order four-potential AL hierarchies, and each member in the
non-isospectral AL hierarchy also possesses two sets of symmetries
which compose a centerless KMV algebra too. These two points make
the algebraic structure of the AL hierarchies quite different from
many other Lax integrable systems. This is due to the symmetrical
form (in terms of $z$ and $z^{-1}$) of the AL spectral matrix (see
\eqref{4p-new}). Besides, it is now clear for the relation between
four-potential and two-potential AL hierarchies. All the even order
members in the four-potential isospectral and non-isospectral
hierarchies can be reduced to two-potential case by directly taking
$(S_n,T_n)=(0,0)$. The new recursion operator for two-potential case
is $L^2$. Meanwhile, the reduction keeps the algebraic structures of
symmetries invariant. The procedure and main results are possible to
apply to Ablowitz-Ladik systems with vector potentials and this will
be discussed elsewhere. This paper will be continued by the second
part in which we will discuss symmetries and recursion operator for
the integrable discrete nonlinear Schr\"odinger equation and the
discrete AKNS hierarchy.

\vspace{1cm}

\section*{Acknowledgement}
This project is supported by the National Natural Science Foundation
of China (10671121) and Shanghai Leading Academic Discipline
Project (No.J50101).

\begin{appendix}
\section{The first few flows and their Lax pairs}\label{A-1}
For the isospectral four-potential AL hierarchy \eqref{hie-iso}, when $m=1$ we have
\begin{align}
u_{n, t_{1}}=K^{(1)}=\left(\begin{array}{cccc}
 (1-Q_{n}R_{n})S_{n}\\-(1-Q_{n}R_{n})T_{n-1}\\ (1-S_{n}T_{n})Q_{n+1}\\ -(1-S_{n}T_{n})R_{n}
 \end{array}\right).
 \end{align}
In its Lax pair
\begin{align}
 V_{n}^{(1)}= \left(
\begin{array}{cc} -\frac{1}{2}Q_{n}T_{n-1}+\frac{1}{2}z^{2} & Q_{n}z\\
T_{n-1}z& \frac{1}{2}Q_{n}T_{n-1}-\frac{1}{2}z^{2}
\end{array}\right).
\end{align}
For $m=-1$, we have
\begin{align}
u_{n, t_{-1}}=K^{(-1)}=\left(\begin{array}{cccc}
 (1-Q_{n}R_{n})S_{n-1}\\-(1-Q_{n}R_{n})T_{n}\\ (1-S_{n}T_{n})Q_{n}\\ -(1-S_{n}T_{n})R_{n+1}
 \end{array}\right),
 \end{align}
and
\begin{align}
 V_{n}^{(-1)}= \left(
\begin{array}{cc} -\frac{1}{2}R_{n}S_{n-1}+\frac{1}{2}z^{-2} & -S_{n-1}z^{-1}\\
-R_{n}z^{-1}& \frac{1}{2}R_{n}S_{n-1}-\frac{1}{2}z^{-2}
\end{array}\right).
\end{align}
For the non-isospectral four-potential AL hierarchy \eqref{hie-non}, when $m=1$ we have
\begin{equation}
u_{n, t_{1}}=\sigma^{(1)}=\left(\begin{array}{cccc}
 (2n+\frac{3}{2})(1-Q_{n}R_{n})S_{n}-Q_{n}(E-1)^{-1}(Q_{n+1}T_{n}+S_{n}R_{n})\\
 -(2n-\frac{1}{2})(1-Q_{n}R_{n})T_{n-1}+R_{n}(E-1)^{-1}(Q_{n+1}T_{n}+S_{n}R_{n})\\
 (2n+\frac{5}{2})(1-S_{n}T_{n})Q_{n+1}-S_{n}(E-1)^{-1}(Q_{n+1}T_{n}+S_{n}R_{n})-S_{n}^{2}R_{n}\\
 -(2n+\frac{1}{2})(1-S_{n}T_{n})R_{n}+T_{n}(E-1)^{-1}(Q_{n+1}T_{n}+S_{n}R_{n})+S_{n}T_{n}R_{n}
 \end{array}\right).
\end{equation}
In its Lax pair
\begin{equation}
V_{n}^{(1)}= \left(
\begin{array}{cc} -(2n+\frac{1}{2})Q_{n}T_{n-1}-(E-1)^{-1}(T_{n-1}Q_{n}+R_{n}S_{n})+nz^{2} & (2n+\frac{1}{2})Q_{n}z\\
(2n-\frac{1}{2})T_{n-1}z& -nz^{2}
\end{array}\right).
\end{equation}
For $m=-1$,
\begin{align}
u_{n, t_{-1}}=\sigma^{(-1)}=\left(\begin{array}{cccc}
 (2n-\frac{1}{2})(1-Q_{n}R_{n})S_{n-1}-Q_{n}(E-1)^{-1}(R_{n+1}S_{n}+T_{n}Q_{n})\\
 -(2n+\frac{3}{2})(1-Q_{n}R_{n})T_{n}+R_{n}(E-1)^{-1}(R_{n+1}S_{n}+T_{n}Q_{n})\\
 (2n+\frac{1}{2})(1-S_{n}T_{n})Q_{n}-S_{n}(E-1)^{-1}(R_{n+1}S_{n}+T_{n}Q_{n})-Q_{n}S_{n}T_{n}\\
 -(2n+\frac{5}{2})(1-S_{n}T_{n})R_{n+1}+T_{n}(E-1)^{-1}(R_{n+1}S_{n}+T_{n}Q_{n})+T_{n}^{2}Q_{n}
 \end{array}\right),
 \end{align}
and
\begin{align}
 V_{n}^{(-1)}= \left(
\begin{array}{cc} nz^{-2} & -(2n-\frac{1}{2})S_{n-1}z^{-1}\\
-(2n+\frac{1}{2})R_{n}z^{-1}&
(2n+\frac{1}{2})R_{n}S_{n-1}+(E-1)^{-1}(S_{n-1}R_{n}+T_{n}Q_{n})-nz^{-2}
\end{array}\right).
\end{align}

\section{Theorem 1 in Ref.\cite{ZDJ-02-JPA}}\label{A-2}
\begin{theorem}
\label{T-B.1} Suppose that the linear problem
\begin{equation}
\Phi_{n+1}=M_n\Phi_n \label{B-1}
\end{equation}
satisfies the following conditions:\\
(1). matrix equation
\begin{equation}
M_n' [X] =(EN_n)M_n -M_nN_n \label{B-2}
\end{equation}
possesses a unique couple of nonzero solutions $X(u(t,n))\in
\mathcal{V}_{s}$ and $N_n=N(z,u(t,n)) \in \mathcal{Q}_{m}(z)$
satisfying
$[M_n,N_n]|_{u_n=0}=0$, \\
(2). for any given $Y(u(t,n))\neq 0 \in \mathcal{V}_{s}$, there
exist solutions $X(u(t,n))\in \mathcal{V}_{s}$ and $N_n \in
\mathcal{Q}_{m}(z)$ satisfying
\begin{equation}
M'_n [X-z^{\alpha}Y]=(EN_n)M_n -M_nN_n ,~ ~ N_n|_{u_n=0}=0,
\end{equation}
where $\alpha$ is a constant related to \eqref{B-1}, then the
following results hold:\\
(1). there exist the following Lax representations for
isospectral flows $\{K\ui{l}\}$:\\
\begin{equation}
M'_n [K\ui{l}]=(EN_n^{(l)})M_n -M_nN_n^{(l)}, ~~l=0,1,\cdots,
\end{equation}
(2). there exists a unique recursion operator $L$  such that
\begin{equation}
u_{n,t}=K\ui{l}=L^l K\ui{0}, ~~l=0,1,\cdots, \label{B-4}
\end{equation}
and $L$ is a strong and hereditary symmetry of each member in the
above hierarchy.
\end{theorem}

\end{appendix}

\end{document}